\documentclass[
  english,
  abstract=true]{scrartcl}

\usepackage{amsmath,amssymb}
\usepackage{hyperref}
\usepackage{enumerate}
\usepackage{amsthm}
\usepackage{cleveref}
\usepackage{siunitx}
\usepackage{tabularx}
\usepackage{booktabs}

\usepackage{iftex}
\ifPDFTeX
  \usepackage[T1]{fontenc}
  \usepackage[utf8]{inputenc}
  \usepackage{textcomp} 
\else 
  \usepackage{unicode-math} 
  \defaultfontfeatures{Scale=MatchLowercase}
  \defaultfontfeatures[\rmfamily]{Ligatures=TeX,Scale=1}
\fi

\usepackage{lmodern}
\ifPDFTeX\else
\fi
\IfFileExists{upquote.sty}{\usepackage{upquote}}{}
\IfFileExists{microtype.sty}{
  \usepackage[]{microtype}
  \UseMicrotypeSet[protrusion]{basicmath} 
}{}
\makeatletter
\@ifundefined{KOMAClassName}{
  \IfFileExists{parskip.sty}{%
    \usepackage{parskip}
  }{
    \setlength{\parindent}{0pt}
    \setlength{\parskip}{6pt plus 2pt minus 1pt}}
}{
  \KOMAoptions{parskip=half}}
\makeatother
\usepackage{xcolor}
\usepackage{graphicx}
\makeatletter
\def\maxwidth{\ifdim\Gin@nat@width>\linewidth\linewidth\else\Gin@nat@width\fi}
\def\maxheight{\ifdim\Gin@nat@height>\textheight\textheight\else\Gin@nat@height\fi}
\makeatother
\setkeys{Gin}{width=\maxwidth,height=\maxheight,keepaspectratio}
\makeatletter
\def\fps@figure{htbp}
\makeatother
\usepackage[style = authoryear, backend=biber, maxbibnames=999, maxcitenames=2,
            doi=true, isbn=false, eprint=false]{biblatex} 
\ifPDFTeX
  \usepackage{libertine}
  \usepackage{libertinust1math}
  \usepackage[scaled=0.83]{beramono}
\else 
  \usepackage{fontspec}
\fi
\usepackage{csquotes}

\ifLuaTeX
  \usepackage{selnolig}  
\fi
\usepackage{bookmark}
\IfFileExists{xurl.sty}{\usepackage{xurl}}{} 
\hypersetup{
    pdftitle={Utility-based optimization of Fujikawa's basket trial design -- Pre-specified protocol of a comparison study},
      hidelinks,
  pdfcreator={LaTeX via pandoc}}
  
\AtEveryBibitem{\clearfield{month}}
\AtEveryCitekey{\clearfield{month}}
\AtEveryBibitem{%
  \clearlist{language}
  \ifentrytype{article}
    {\clearfield{url}\clearfield{urlyear}\clearfield{note}}
    {}%
  \ifentrytype{unpublished}
    {\clearfield{url}\clearfield{urlyear}}
    {}%
}
\crefname{thm}{theorem}{theorems}
\crefname{defn}{definition}{definitions}
\crefname{lem}{lemma}{lemmas} 
\crefname{cor}{corollary}{corollaries} 
\crefname{conj}{conjecture}{conjectures}
\crefname{prop}{proposition}{propositions}
\crefname{alg}{algorithm}{algorithms}
\theoremstyle{definition}
\newtheorem{thm}{Theorem}[section]

\newtheorem{prop}[thm]{Proposition}

\title{Utility-based optimization of Fujikawa's basket trial design --
Pre-specified protocol of a comparison study}
\date{May 17, 2024}

\usepackage[noblocks]{authblk}

\author[1] {Lukas D Sauer}
\author[2] {Alexander Ritz}
\author[1] {Meinhard Kieser}
\affil[1] {Institute of Medical Biometry, Heidelberg University,
Germany}
\affil[2] {Institute of Mathematics, Clausthal University of Technology,
Germany}

\ifPDFTeX
  \usepackage{upgreek}
  \newcommand\symbfupphi{\boldsymbol{\upphi}}
\else 
  \newcommand\symbfupphi{\symbfup{\phi}}
\fi
\DeclareMathOperator{\JSD}{JSD}
\DeclareMathOperator{\KLD}{KLD}
\DeclareMathOperator{\HLD}{HLD}

\usepackage{booktabs}
\usepackage{caption}
\usepackage{longtable}
\usepackage{colortbl}
\usepackage{array}

\begin{document}

\maketitle

\begin{abstract}
Basket trial designs are a type of master protocol in which the same
therapy is tested in several strata of the patient cohort. Many basket
trial designs implement borrowing mechanisms. These allow sharing
information between similar strata with the goal of increasing power in
responsive strata while at the same time constraining type-I error
inflation to a bearable threshold. These borrowing mechanisms can be
tuned using numerical tuning parameters. The optimal choice of these
tuning parameters is subject to research. In a comparison study using
simulations and numerical calculations, we are planning to investigate
the use of utility functions for quantifying the compromise between
power and type-I error inflation and the use of numerical optimization
algorithms for optimizing these functions. The present document is the
protocol of this comparison study, defining each step of the study in
accordance with the ADEMP scheme for pre-specification of simulation
studies.
\end{abstract}

\section{Introduction}\label{introduction}

With the dawn of precision medicine and targeted antibody therapies, the
wish for more flexible trial designs compared to randomized controlled
trials has been emphasized in both clinical research and methodology.
While randomized controlled trials are still the gold standard of
clinical research owing to their high internal validity, in some
contexts they may be costly, slow, unethical, or simply not feasible.
The term \emph{master protocols} summarizes more flexible trial designs
that combine features such as the addition of promising new treatment
arms, the removal of ineffective treatment arms, and the combination of
arms with different endpoints. A commonly requested idea is testing the
same treatment in several substrata of a patient cohort. Such a master
protocol, i.e., a design that unifies several strata in a single
clinical trial, is called \emph{basket trial design}. This unification
streamlines the planning phase of the different strata and parallelizes
their recruitment, resulting in an efficient use of resources. Basket
trial designs are beginning to be picked up in practice. A systematic
literature review conducted on February 20, 2023, in MEDLINE, Embase,
and the Cochrane Central Register of Controlled Trials found 146
oncology-related basket trials \parencite{kasim_basket_2023}. Especially
in early stages of clinical research, these designs also offer
statistical benefits: So-called \emph{borrowing} techniques allow strata
with similar responses to the treatment to share information with one
another, thereby leveraging power despite small sample sizes while
keeping type-I error rates only moderately inflated. The earliest
publication known to us that suggests borrowing between strata was
published in 2003 \parencite{thall_hierarchical_2003} and the last ten
years showed a colorful bouquet of Bayesian and frequentist borrowing
mechanisms being introduced to methodological research
\parencite[see][]{pohl_categories_2021}.

Borrowing usually depends on the choice of several tuning parameters. An
optimal choice of tuning parameters has to offer a compromise between
multiple components: achieving high power in responsive strata, keeping
type-I error rates low in unresponsive strata and maintaining a good
balance between these two measures across several response scenarios.
This compromise can be quantified using utility functions. The optimal
choice of tuning parameters can then be found by finding optimal utility
function values with the help of numerical optimization algorithms. In
the following document, we present the protocol of a comparison study
planned to investigate utility-based optimization of basket trial
designs using both simulations and numerical calculations. The basket
trial design that we are considering as an example is a Bayesian design
introduced by \cite{fujikawa_bayesian_2020}.

\section{Methodology of utility functions in basket trial
designs}\label{methodology-of-utility-functions-in-basket-trial-designs}

In the statistical planning of clinical trials, the communication of
type-I error rate (TOER) and power to stakeholders such as principal
investigators, sponsors, ethical committees and regulatory authorities
is essential. TOER is the probability of rejecting the null hypothesis
conditional on the assumption that the null hypothesis is true. Power is
the probability of rejecting the null hypothesis conditional on the
assumption that some alternative hypothesis of interest holds. While
these measures are purely frequentist in nature, they may also be
requested during the planning of Bayesian trial designs.

In the planning of basket trial designs, this demand for TOER and power
calculation is confronted with several challenges. Firstly, every
stratum may have its own null and alternative hypothesis of interest so
that a multitude of combinations of null and alternative hypotheses
across baskets can be considered. Secondly, control of TOER in a
scenario may not be possible if we want to employ borrowing in order to
leverage power. \cite{kopp-schneider_power_2020} proved that power
increase always comes at the cost of TOER inflation in the context of
borrowing from external data sources whenever a uniformly most powerful
test exists. While we are not aware of a formal transfer of their
argument to the context of basket trials, it is plausible that the
argument holds in that setting as well.

In communication with stakeholders, the best practice may be to
communicate both TOER and power for each stratum across a range of
plausible scenarios combining null and alternative hypotheses in
different strata.

However, when searching for the optimal choice of tuning parameters of a
basket trial designs we need to combine TOER and power across strata and
scenarios. A natural way of combining these values is by defining an
appropriate utility function. Then, optimization algorithms can be
employed in order to find the optimal tuning parameter vector.
Constraints, e.g.~the maximally tolerated TOER values, may either be
incorporated into the utility function allowing for unconstrained
optimization methods or may be set up as separate inequalities, asking
for constrained optimization approaches. This utility-based optimization
approach has already been employed in the context of adaptive designs
\parencite{pilz_optimal_2021} and in the context of planning several
stages of drug development \parencite{kirchner_utilitybased_2016} even
though formal control of TOER is possible in these contexts. In the
methodology of basket trial designs, we are aware of a first
utility-based approach presented in \cite{jiang_optimal_2021}. Their
approach will also be considered in the presented comparison study
protocol and will be supplemented by several other optimization
algorithms and utility functions.

\section{Goals of this comparison study}

This comparison study's goal is to find optimal tuning parameter
combinations for Fujikawa's basket trial design in a general framework
that could subsequently be generalized to other basket trial designs as
well. In particular, the study is divided into three parts addressing
three questions related to finding optimal parameter combinations:

\begin{enumerate}[I.]
\item Which type of optimization algorithm for finding the optimal tuning parameter vector $\symbfupphi^*$ should be preferred in terms of runtime and reliability?
\item What is a good definition of \emph{optimal} tuning parameter vector $\symbfupphi^*$  that takes the desire for maximizing the detection probability of active strata as well constraints on TOER into account while delivering favorable results across a range of outcome scenarios? This question amounts to finding an appropriate utility function. We will use the best algorithm found in Part I, and apply it to a variety of different utility functions.
\item How does the optimal tuning parameter vector $\symbfupphi^*$  found as a result of parts I and II perform in comparison to the tuning parameter combinations suggested in \cite{fujikawa_bayesian_2020}?
\end{enumerate}

The document will begin with a brief introduction into basket trial
designs in general and into Fujikawa's design in particular. Afterwards,
we will explain the true scenarios of interest in our study as well as
the considered utility function and optimization algorithms. Then we
will present the pre-specified plan of the three parts of the study
structured using the ADEMP scheme by \cite{morris_using_2019}. Part I is
a simulation study for comparing different algorithms, parts II and III
are comparison studies where all measures of interest can be calculated
exactly.

\section{Basket trial designs}

A basket trial design is a clinical trial design used primarily in
oncological single-arm phase II studies. It tests the same null
hypothesis in several strata. In the literature, either the ensemble of
all strata together is called \enquote{basket} or the strata themselves
are called \enquote{baskets}. In the following, the endpoint will always
be binary. While this could be any binary endpoint, we will without loss
of generality refer to response to a treatment vs.~no response
throughout the text. Strata that have a sufficiently high true response
rate to the treatment are called \emph{active}, otherwise they are
called \emph{inactive}.

Consider the design from \cite{fujikawa_bayesian_2020}, which is based
on an alteration of the beta-binomial model. Using the notation from
\cite{pohl_categories_2021}, it can be defined as follows: Let \(n_i\)
resp. \(r_i\) be the number of patients resp. responders in stratum
\(i\in 1, \ldots, I\) for some number of strata \(I\). The sampling
distribution is simply the binomial distribution
\[r_i \sim \mathrm{Bin}(n_i, p_i),\] where the true rate \(p_i\) follows
a prior beta distribution with shape parameters \(a_i, b_i > 0\),
\[p_i \sim \mathrm{Beta}(a_i, b_i).\] For analysis of data, Fujikawa et
al.~recommend an uninformative choice of prior distribution. We choose
\(a_i=b_i=1\). By conjugacy, the usual posterior distribution would be
\begin{equation}
\label{eqn:beta-posterior}
p_i \sim \mathrm{Beta}(a_i+r_i, b_i+n_i-r_i)=\mathrm{Beta}_i^{\text{post}}.
\end{equation} Now this posterior is altered by introducing a borrowing
mechanism, \begin{equation}
\label{eqn:borrowing-posterior}
p_i \sim \mathrm{Beta}(\textstyle\sum_j\omega_{ij}\cdot(a_j+r_j),\textstyle\sum_j\omega_{ij}\cdot(b_j+n_j-r_j))=\mathrm{Beta}_i^{\text{bor}}.
\end{equation} Here \(\omega_{ij}\) is a similarity measure defined by
\(\omega_{ij} = \mathbf{1}(\tilde\omega_{ij}^\varepsilon>\tau)\cdot\tilde\omega_{ij}^\varepsilon\),
where we have \(\tau\in[0,1]\), \(\varepsilon \geq 0\), and
\(\tilde\omega_{ij}= 1 - \JSD(\mathrm{Beta}_i^{\text{post}}, \mathrm{Beta}_j^{\text{post}})\)
with the Jensen-Shannon divergence \(\JSD\) of the unaltered
beta-binomial posterior distributions from
\Cref{eqn:beta-posterior}.\footnote{Note that in the original publication \parencite{fujikawa_bayesian_2020}, the bound for $\varepsilon$ is set to $\varepsilon\geq 1$. However, there is no mathematical or design-related reason to not allow values between 0 and 1.}
\(\JSD\) is a measure of divergence of probability distribution which
implies that \(\omega_{ij}\) becomes a measure of similarity of
probability distribution, which is set to \(0\) if the similarity is
less or equal to \(\tau\). The Jensen-Shannon divergence is defined as
\[ \JSD(P, Q)=\frac{1}{2}(\KLD(P, M) + \KLD(Q, M)),\] where
\(M=\frac{1}{2}(P+Q)\) is the mixture distribution of \(P\) and \(Q\)
\parencite[see][]{fujikawa_bayesian_2020}. Here, \(\KLD\) is the
Kullback-Leibler divergence defined as
\[ \KLD(P, Q)=\int_{\mathcal{X}}P(x)\log\left(\frac{P(x)}{Q(x)}\right)\mu(dx), \]
where \((\mathcal X, S, \mu)\) is the probability space on which \(P\),
\(Q\) or \(M\) are defined \parencite[see][]{kullback_information_1951}
and \(P(x)\) denotes the Radon-Nikodym derivative with respect to
\(\mu\), i.e.~the probability density function in case \(\mu\) is chosen
to be the Lebesgue measure. Here, \(\log(\cdot)\) denotes the natural
logarithm. We use this logarithm for better comparability with
Fujikawa's results where the natural logarithm is used as well. In
\cite{baumann_basket_2023}, the logarithm with base 2 is used as it
implies that the Jensen-Shannon divergence ranges from 0 to 1.

The test decision whether stratum \(i\) is \emph{detected} as active is
based on the posterior probability of lying above a desired target rate
\(p_i^*\), i.e., \begin{equation}
\label{eq:detect-active}
P(p_i>p^*_i|\mathbf{r})\geq \lambda,
\end{equation} where \(\mathbf{r}=(r_i)_i\) is the vector of responses
and where the posterior probability \(P(\cdot|\mathbf{r})\) is defined
with respect to the borrowing posterior
\(\mathrm{Beta}_i^{\text{bor}}\).

We denote by \(\symbfupphi = (\lambda, \varepsilon, \tau)\) the vector
of tuning parameters. In the examples of \cite{fujikawa_bayesian_2020},
the shape parameter \(\varepsilon\), the similarity cutoff \(\tau\) and
the detection threshold \(\lambda\) are chosen to be either
\(\symbfupphi = (\lambda, \varepsilon, \tau) = (0.99, 2, 0)\) or
\((0.99, 2, 0.5)\), but it is unclear whether this is the optimum for
choosing the tuning parameters.

Concerning the relationship between \(\varepsilon\) and \(\tau\), one
should note two things. Firstly, they are redundant when it comes to
defining the minimal similarity for which borrowing is still allowed.
Indeed, for a given minimal similarity \(\tilde\omega^*\) and a given
\(\tau\), we can choose
\(\varepsilon_{\tilde\omega^*}(\tau)=\log_{\tilde\omega^*}(\tau)\) (and
analogously
\(\tau_{\tilde\omega^*}(\varepsilon)=(\tilde\omega^*)^\varepsilon\))
such that the function
\(\omega_{ij} = \mathbf{1}(\tilde\omega_{ij}^\varepsilon>\tau)\cdot\tilde\omega_{ij}^\varepsilon\)
\(\max_{r_k\neq r_l}\tilde\omega_{kl}\) is greater than 0 if and only if
\(\tilde\omega_{kl}>\tilde\omega^*\).

Secondly, this implies that when choosing
\(\tilde\omega^* = \max_{r_i\neq r_j}\tilde\omega_{i,j}\) or greater,
then borrowing is only allowed for two substrata that have identical
response rates. We call this the \emph{extreme borrowing boundary}
\(\varepsilon_{\text{extreme}}(\tau)\). Increasing \(\varepsilon\) or
\(\tau\) above this boundary does not change the behavior as
\(\omega_{ij}\) is either equal to \(1\) in case \(r_i=r_j\) or equal to
\(0\) in case \(r_i\neq r_j\).

The parameter space upon which we will perform optimization in this
study will be \(\lambda\in[0, 1]\), \(\varepsilon\in[0, \infty)\) and
\(\tau\in[0, 1]\). For the grid search algorithm (explained below), we
will restrict ourselves to \(\varepsilon\in[0, 25]\). This means that
the parameter space will encompass the parameter suggestions by Fujikawa
et al.~and that we can study the behavior of the extreme borrowing
boundary for higher \(\tau\) values (approx. above \(\tau=0.6\)) as can
be seen in \Cref{fig:extreme-borrowing}.

\begin{figure}
\centering
\includegraphics[alt={The figure shows the extreme borrowing boundary of epsilon on the y axis plotted against tau on the x axis.}]{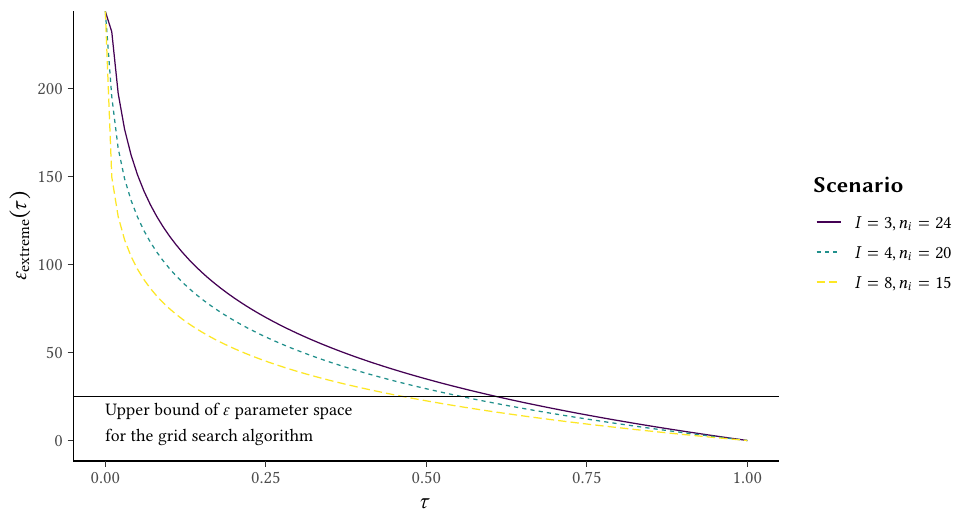}
\caption{\label{fig:extreme-borrowing}Extreme borrowing boundary $\varepsilon_{\text{extreme}}( \tau )$ for different stratum counts $I$ and per-stratum patient counts $n_i$ and the upper bound of the $\varepsilon$ parameter space for the grid search algorithm}
\end{figure}

\section{Outcome scenarios}
\label{sec:scenarios}

In order to study the performance of the tuning parameter combinations,
each step of the comparison study will consider the scenarios following
in the list below. Firstly, we are using scenarios from related
methodological publications, namely \cite{fujikawa_bayesian_2020} whose
design we are studying, \cite{baumann_basket_2023} whose underlying R
package we are using for calculation of performance measures, and
\cite{krajewska_new_2021} in anticipation of a future comparison of
frequentist and Bayesian sharing techniques. The number of strata in
these methodological publications vary between \(I=3\), 4 and 8 and the
per-stratum sample sizes range from \(n_i=15\) to \(n_i=24\). The
scenarios with low stratum and patient counts
\((I=3, n_i=24, p_0 = 0.2)\) and \((I=4, n_i=0, p_0 = 0.15)\) enable
exact calculation of performance measures (up to numerical precision of
integration methods) without the necessity for Monte Carlo methods.

In order to supplement these methodologically motivated scenarios by
clinically realistic scenarios, we considered the publications and
published supplementary material of four different systematic literature
reviews of planned and ongoing basket trials.
\cite{park_systematic_2019} and \cite{meyer_evolution_2020} investigated
planned and ongoing master protocols including basket trials in July
2019 and January 2020, respectively. While their search scopes and the
identified trials are partially overlapping, we considered both reviews:
\citeauthor{park_systematic_2019} found more basket trials than Meyer et
al., but on the other hand, \citeauthor{meyer_evolution_2020} presented
invidivual trial characteristics in the main text which made information
easy and safe to extract. \cite{kasim_basket_2023} mirrored
\citeauthor{park_systematic_2019}'s search strategy in February 2023,
focussing specifically on basket trials in oncology. Concerning search
date, it is the most up-to-date of the considered reviews. The umbrella
review by \cite{haslam_umbrella_2023} again searched basket trials in
oncology, albeit about a year earlier in March 2022. We still considered
their review as they report summary data of response rates in the
strata.

In \cref{tab:sys-rev-basket}, we summarized trial characteristics as
presented in the four systematic reviews. Median per-group sample size
was higher in \cite{meyer_evolution_2020} compared to
\cite{park_systematic_2019} and \cite{kasim_basket_2023}. As
\citeauthor{meyer_evolution_2020} reported the planned sample size, we
deem the median per-group sample size in \cite{park_systematic_2019} and
\cite{kasim_basket_2023} more realistic. Considering these summary data,
it is apparent that the total sample size and per-stratum sample size
are usually higher in actual clinical trials than in the scenarios taken
from above-mentioned methodological publications. Hence, we included
four more scenario sets that are similar to actual completed basket
trials registered on \url{ClinicalTrials.gov}. These four scenario sets
represent trials with medium and large total sample sizes and medium and
large per-stratum sample sizes.

Considering response rates, the active response rates in the
methodological publications are slightly larger than the 75\%-quantile
reported in the systematic review by \citeauthor{park_systematic_2019}.
The four realistic scenario sets complement these response rates by
smaller response rates and smaller effect sizes.

\begin{table}[p]
\caption{Characteristics of basket trials as found in the literature. Per-group no. of patients was calculated as the total number of patients divided by the number of subgroups. Data were taken from the publications of \cite{meyer_evolution_2020}, Table 1, and \cite{haslam_umbrella_2023}, Table 1, and from the supplementary online material  of \cite{park_systematic_2019}, Additional file 1, Tables S9 and S10, and of \cite{kasim_basket_2023}, Data Sheet 2.}
\label{tab:sys-rev-basket}
\centering

\setlength{\LTpost}{0mm}
\begin{longtable}{lcccc}
\toprule
 & \multicolumn{4}{c}{\textbf{Systematic review}} \\ 
\cmidrule(lr){2-5}
\textbf{Characteristic} & \textbf{Park 2019} & \textbf{Meyer 2020}\textsuperscript{\textit{1}} & \textbf{Haslam 2023}\textsuperscript{\textit{2}} & \textbf{Kasim 2023}\textsuperscript{\textit{3}} \\ 
\midrule\addlinespace[2.5pt]
No. of trials, N & 49 & 16 & 25 & 146 \\ 
Total no. of patients &  &  &  &  \\ 
    Median (25–75\%) & 205 (90–500) & 101 (95–481) & 48 (30–122) & 70 (35–145) \\ 
    Range & 12–6,452 & 71–11,000 &  & 0–1,609 \\ 
    Unknown & 1 & 1 &  & 0 \\ 
No. of strata &  &  &  &  \\ 
    Median (25–75\%) & 5.0 (3.5–6.5) & 3.5 (1.8–4.0) &  & 3.0 (1.0–5.0) \\ 
    Range & 2.0–27.0 & 1.0–7.0 &  & 1.0–24.0 \\ 
    Unknown & 14 & 0 &  & 115 \\ 
Per-stratum no. of patients &  &  &  &  \\ 
    Median (25–75\%) & 43 (30–71) & 93 (23–145) &  & 38 (18–61) \\ 
    Range & 17–514 & 14–2,750 &  & 6–172 \\ 
    Unknown & 14 & 1 &  & 115 \\ 
Response rate [\%] &  &  &  &  \\ 
    Median (25–75\%) &  &  & 23.1 (8–30) &  \\ 
    Unknown & 49 & 16 &  & 146 \\ 
\bottomrule
\end{longtable}
\begin{minipage}{\linewidth}
\textsuperscript{\textit{1}}Meyer et al. 2020 reported the planned sample size and planned number of subgroups.\\
\textsuperscript{\textit{2}}Only summarized data were available for Haslam et al. 2023.\\
\textsuperscript{\textit{3}}Kasim et al. 2023 reported the number of subgroups only for completed basket trials.\\
\end{minipage}

\end{table}

\begin{enumerate}
\item Three-stratum scenarios from \cite{fujikawa_bayesian_2020}: The number of strata is $I=3$, strata with true response rate $p_i=0.2$ are truly inactive, strata with $p_i>0.2$ are truly active. We will consider sample sizes of $n_i=24$ per stratum and the following combinations of true response rates:
\begin{enumerate}
\item 0 of 3 active strata: $\mathbf{p}=(0.2,0.2,0.2)$,
\item 1 of 3 active strata: $\mathbf{p}=(0.2,0.2,0.5)$,
\item 2 of 3 active strata: $\mathbf{p}=(0.2,0.5,0.5)$, and
\item 3 of 3 active strata: $\mathbf{p}=(0.5,0.5,0.5)$.
\end{enumerate}
\item Four-stratum scenarios from \cite{baumann_basket_2023}: The number of strata is $I=4$, strata with true response rate $p_i=0.15$ are truly inactive, strata with $p_i>0.15$ are truly active. We will consider sample sizes of $n_i=20$ per stratum and the following combinations of true response rates:
\begin{enumerate}
\item 0 of 4 active strata: $\mathbf{p}=(0.15, 0.15, 0.15, 0.15)$,
\item 1 of 4 active strata: $\mathbf{p}=(0.15, 0.15, 0.15, 0.4)$,
\item 2 of 4 active strata $\mathbf{p}=(0.15, 0.15, 0.4, 0.4)$,
\item 3 of 4 active strata $\mathbf{p}=(0.15, 0.4, 0.4, 0.4)$,
\item 4 of 4 active strata $\mathbf{p}=(0.4, 0.4, 0.4, 0.4)$,\footnote{In \cite{baumann_basket_2023}, these first five scenarios are called global null, good nugget, half, bad nugget and global alternative.}
\item one in the middle: $\mathbf{p}=(0.4, 0.4, 0.3, 0.5)$, and
\item linear: $\mathbf{p}=(0.15, 0.25, 0.35, 0.45)$.
\end{enumerate}
\item Eight-basket scenarios analogous to \cite{krajewska_new_2021}: The number of strata is $I=8$, strata with true response rate $p_i=0.15$ are truly inactive, strata with $p_i>0.15$ are truly active. We will consider sample sizes of $n_i=15$ per stratum. This is different from the exact sample size in  \cite{krajewska_new_2021}, where the sample size depends on a clustering decision made at interim. However, a sample size of $n_i=15$ is the worst case achieved when the clustering algorithm separates all strata. We will consider the following combinations of true response rates:
\begin{enumerate}[a) {-- i)}]
\setlength{\itemindent}{1em}
\item $a$ of 8 active strata: $\mathbf{p}=(0.15,\ldots, 0.15, \underbrace{0.45,\ldots, 0.45}_a)$ with $0\leq a\leq 8$.
\end{enumerate}
\item Scenario with medium total sample size and small effect sizes, similar to NCT02454972: In Table S5 of the supplementary Data Sheet 1 to \citeauthor{kasim_basket_2023}, the basket trial with above-mentioned \url{ClinicalTrials.gov} ID is summarized. It had $I=9$ strata with per-stratum sample sizes ranging from 13 to 105 (median: 23) and observed per-stratum response rates of $(0.056, 0.000, 0.113, 0.143, 0.043, 0.000, 0.286, 0.065, 0.362)$. Most strata had assumed null response rates of 0.01 (corresponding to standard of care) and targeted response rates of 0.1. In our comparison study, we will consider a simplified version of this trial: We will consider a trial with $I=9$ strata with $n_i = 23$ patients each. Strata with $p_i=0.01$ will be considered inactive, strata with $p_i > 0.01$ will be considered truly active. For optimization, we will then consider the following data scenarios:
\begin{enumerate}[a) {-- j)}]
\setlength{\itemindent}{1em}
\item $a$ of 9 active strata: $\mathbf{p}=(0.01,\ldots, 0.01, \underbrace{0.10,\ldots, 0.10}_a)$ with $0\leq a\leq 9$.
\end{enumerate}
When calculating the performance measures of the optimized parameter combination in part II of the comparison study, we will also consider performance measures assuming that the observed response rates from the clinical trial are the true response rates.
\item Scenario with large total sample size and large number of baskets, similar to NCT02054806: In \citeauthor{kasim_basket_2023}'s supplementary material, another  basket trial is summarized. It has  $I=20$ strata with per-stratum sample sizes ranging from 16 to 27 (median: 24) and assumed null and target response rates of 0.10 and 0.35, respectively. We will consider a simplified version with $I=20$ strata with $n_i=24$ patients each in which strata with $p_i=0.10$ and $p_i>0.10$ will be considered inactive and active, respectively. For optimization, we will then consider the following data scenarios:
\begin{enumerate}[a) {-- k)}]
\setlength{\itemindent}{1em}
\item $a$ of 20 active strata: $\mathbf{p}=(0.10,\ldots, 0.10, \underbrace{0.35,\ldots, 0.35}_a)$ with $a = 0, 2, 4, \ldots, 20$. 
\end{enumerate}
We will only consider even numbers in order to halve computation time. Again, we will also investigate the observed response rates
\begin{align*}(&0.160, 0.174, 0.120, 0.120, 0.167, 0.043, 0.130, 0.304, 0.080, 0.042,\\
&0.200, 0.259, 0.063, 0.115, 0.000, 0.174, 0.115, 0.333, 0.091, 0.056)
\end{align*}
reported for the actual clinical trial when calculating the performance measures.
\item Scenario with medium per-stratum sample size similar to NCT01848834: In \citeauthor{kasim_basket_2023}'s supplementary material, yet another clinical trial is summarized. It has $I=4$ strata with per-stratum sample sizes ranging from 32 to 60 (median: 36) and assumed null and target response rates of 0.10 and 0.35 in half of the baskets. We will consider a simplified version with $I=4$ strata with $n_i=36$ patients each in which strata with $p_i=0.10$ and $p_i>0.10$ will be considered inactive and active, respectively. For optimization, we will then consider the following data scenarios:
\begin{enumerate}[a) {-- e)}]
\setlength{\itemindent}{1em}
\item $a$ of 4 active strata: $\mathbf{p}=(0.10,\ldots, 0.10, \underbrace{0.35,\ldots, 0.35}_a)$ with $0\leq a \leq 4$.
\end{enumerate}
As before, when calculating the performance measures we will also investigate the observed response rates $(0.156, 0.167, 0.212, 0.205)$ reported for the actual clinical trial assuming that they were identical to the true response rates .
\item Scenario with large per-stratum sample size similar to NCT01631552: \citeauthor{kasim_basket_2023} also included the study with above-mentioned ID, whose details can be found on the website \url{ClinicalTrials.gov}, see \cite{gileadsciences_study_2024} in the references. The response to treatment was observed separately in $I=3$ strata in 45, 54 and 108 patients, respectively (median: 54). Assumed response rates used in the planning phase were neither available at \url{ClinicalTrials.gov} nor in the material provided by \citeauthor{kasim_basket_2023} Hence, we assumed a null response rate of 0.15 and a target response rate of 0.30. In particular, we will consider a simplified version with $I=3$ strata with $n_i=54$ per stratum in which strata with $p_i=0.15$ and $p_i>0.15$ will be considered inactive and active, respectively. For optimization, we will then consider the following data scenarios:

\begin{enumerate}[a) {-- d)}]
\setlength{\itemindent}{1em}
\item $a$ of 3 active strata: $\mathbf{p}=(0.15,\ldots, 0.15, \underbrace{0.30,\ldots, 0.30}_a)$ with $0\leq a \leq 3$.
\end{enumerate}

We will again investigate the observed response rates $(0.289, 0.315, 0.333)$ as a further data scenario when calculating the performance measures.
\end{enumerate}

\section{Utility functions}
\label{sec:ut-fun}

Under a given true scenario \(\mathbf{p}\) and a tuning parameter vector
\(\symbfupphi\), we define the power in a truly active stratum \(i\) as
the probability
\[ \mathrm{pow}_i(\symbfupphi,\mathbf{p})=P(i\text{ detected}|\symbfupphi,\mathbf{p}, i\text{ active})\]
and the type-I error rate (TOER) in a truly inactive stratum \(i\) as
the probability
\[ \mathrm{toer}_i(\symbfupphi,\mathbf{p})=P(i\text{ detected}|\symbfupphi,\mathbf{p}, i\text{ inactive}).\]
In this context, all probabilities are \enquote{frequentist}, meaning
that they are defined with respect to true binomial sampling
distributions \(r_i \sim \mathrm{Bin}(n_i, p_i)\) for all \(i\) without
being modeled on a prior distribution. However, the test decision
whether a stratum is detected to be active is made according to the
\enquote{Bayesian} borrowing posterior probability as defined in
\Cref{eq:detect-active}. This approach is sometimes called
\enquote{pragmatic Bayesianism}.

Let \(R\subseteq \{1,\ldots,I\}\) be the set of active strata with
respect to \(\mathbf{p}\), \(R^c\subseteq \{1,\ldots,I\}\) the set of
inactive strata. Then, we define the experiment-wise power (EWP) as the
probability
\[ \mathrm{ewp}(\symbfupphi,\mathbf{p}) = P(\exists i\in R: i\text{ detected}|\symbfupphi,\mathbf{p}).\]
Analogously, we define the family-wise error rate (FWER) as the
probability
\[ \mathrm{fwer}(\symbfupphi,\mathbf{p}) = P(\exists i\in R^c: i\text{ detected}|\symbfupphi,\mathbf{p}).\]
Finally, we define the expected number of correct decisions (ECD) as
\[ \mathrm{ecd}(\symbfupphi,\mathbf{p}) =  \sum_{i\in R}P(i\text{ detected}|\symbfupphi,\mathbf{p}) + \sum_{i\in R^c}P(i\text{ not detected}|\symbfupphi,\mathbf{p}).\]
Based on these functions, we define the following utility functions.
Across all definitions, \(\xi_1, \xi_2>0\) are penalty parameters set to
1 by default.

\begin{enumerate}
\item \label{it:u-ewp} Discontinuous family-wise power-error function \[u_{\text{ewp}}(\symbfupphi,\mathbf{p_1},\mathbf{p_2})=\begin{cases}\mathrm{ewp}(\symbfupphi,\mathbf{p_1}) & \text{if } \mathrm{fwer}(\symbfupphi,\mathbf{p_2}) < \eta_1,\text{ and}\\
      -\xi_1\cdot\mathrm{fwer}(\symbfupphi,\mathbf{p_2}) & \text{if } \mathrm{fwer}(\symbfupphi,\mathbf{p_2}) \geq \eta_1,\end{cases}\]
usually with $\mathbf{p_2}$ being the global null scenario and $\eta_1=0.05$.
\item  \label{it:u-ecd} Expected number of correct decisions \[u_{\text{ecd}}(\symbfupphi,\mathbf{p_1},\mathbf{p_2}) = \begin{cases}\mathrm{ecd}(\symbfupphi,\mathbf{p_1}), & \text{if } \mathrm{fwer}(\symbfupphi,\mathbf{p_2}) < \eta_1,\text{ and}\\
-\xi_1\cdot\mathrm{fwer}(\symbfupphi,\mathbf{p_2}) & \text{if } \mathrm{fwer}(\symbfupphi,\mathbf{p_2}) \geq\eta_1,\end{cases}\]
usually with $\mathbf{p_2}$ being the global null scenario.
\item  \label{it:u-2ewp} Two-level family-wise power-error function \[u_{\text{2ewp}}(\symbfupphi,\mathbf{p}) = \mathrm{ewp}(\symbfupphi,\mathbf{p}) - \left(\xi_1\mathrm{fwer}(\symbfupphi,\mathbf{p}) + \xi_2(\mathrm{fwer}(\symbfupphi,\mathbf{p}) - \eta_2)\mathbf 1(\mathrm{fwer}(\symbfupphi,\mathbf{p}) - \eta_2) \right),\]
where $\eta_2\in[0,1]$ is a threshold for imposing harder FWER penalty, set to 0.1 by default.
\item \label{it:u-2pow} Two-level stratum-wise power-error function \begin{align*} u_{\text{2pow}}(\symbfupphi,\mathbf{p}) = \sum_{i\in R}\mathrm{pow}_i(\symbfupphi,\mathbf{p}) - \sum_{j\in R^c}&(\xi_1\mathrm{toer}_j(\symbfupphi,\mathbf{p}) \\&
+ \xi_2(\mathrm{toer}_j(\symbfupphi,\mathbf{p}) - \eta_2)
\mathbf 1(\mathrm{toer}_j(\symbfupphi,\mathbf{p}) - \eta_2) ),\end{align*}
as suggested in \cite{jiang_optimal_2021}.

\item Scenario-averaged versions of the above utility functions
\[ \bar u_l(\symbfupphi,\mathbf{p_2})=\sum_{\mathbf p\in\{\mathbf{p},\ldots\}}w_\mathbf{p}u_l(\symbfupphi, \mathbf{p},\mathbf{p_2})\quad\text{for }l=\text{ewp},\text{ecd (i.e.\ from \Cref{it:u-ewp} and \ref{it:u-ecd})},\]
and
\[ \bar{u_l}(\symbfupphi)=\sum_{\mathbf p\in\{\mathbf{p},\ldots\}}w_\mathbf{p}u_l(\symbfupphi, \mathbf{p})\quad\text{for }l=\text{2ewp},\text{2pow (i.e.\ from \Cref{it:u-2ewp} and \ref{it:u-2pow})},\]

where $\{\mathbf{p},\ldots\}$ is a set of scenarios of interest, e.g., the set of scenarios with number of strata $I$ from \Cref{sec:scenarios}, and $w_\mathbf{p}$ are weights with $\sum_\mathbf{p} w_\mathbf{p} = 1$, e.g., $w_\mathbf{p}=\frac{1}{\#\{\mathbf{p},\ldots\}}$ for all $\mathbf{p}$. 

\item Scenario-averaged utility functions with penalty of maximal TOER inflation
\[ \bar u_{l,\text{pen}}(\symbfupphi,\mathbf{p_2})= 
\begin{cases}\bar u_l(\symbfupphi,\mathbf{p_2}) & \text{if }\max_{\mathbf p, j} \mathrm{toer}_j(\symbfupphi,\mathbf{p}) < \eta_3,\\
-\xi_3\cdot\max_{\mathbf{p}, j} \mathrm{toer}_j(\symbfupphi,\mathbf{p})& \text{if }\max_{\mathbf p, j} \mathrm{toer}_j(\symbfupphi,\mathbf{p}) \geq \eta_3,
\end{cases}\]
for $l=$ ewp or ecd with the maximum defined across all $\mathbf p\in\{\mathbf{p},\ldots\}$ and $ j\in R^c$, and
\[ \bar{u}_{l,\text{pen}}(\symbfupphi)=\begin{cases}\bar u_l(\symbfupphi) & \text{if }\max_{\mathbf p, j} \mathrm{toer}_j(\symbfupphi,\mathbf{p}) < \eta_3,\\
-\xi_3\cdot\max_{\mathbf p, j} \mathrm{toer}_j(\symbfupphi)& \text{if }\max_{\mathbf p, j} \mathrm{toer}_j(\symbfupphi,\mathbf{p}) \geq \eta_3,
\end{cases} \]
for $l=$ 2ewp or 2pow, where we choose $\eta_3=0.2$  and $\xi_3=1000$. In order to make the penalty work, $\xi_3\cdot\eta_3$ should be greater than the absolute value $|\min \bar u_l |$.
\end{enumerate}

We call the first four functions single-scenario utility functions. The
scenario-averaged two-level stratum-wise power-error utility function
\(\bar u_\text{2pow}\) was suggested in \cite{jiang_optimal_2021},
except for the fact that the authors use all possible partitions with
respect to response rates whereas we allow arbitrary scenario sets. The
scenario-averaged expected number of correct decisions function
\(\bar u_\text{2pow}\) emulates the optimization algorithm from
\cite{baumann_basket_2023} which in turn took the algorithm from
\cite{broglio_comparison_2022}. There, the detection threshold
\(\lambda\) is first optimized to keep FWER with respect to the global
null scenario below a threshold and then the mean ECD across all
scenarios is maximized subsequently.

The scenario-averaged utility functions with penalty of maximal TOER
inflation are promising as they present a good compromise between
research goals and regulatory requirements. In the context of borrowing
from an external data source, it is an established fact that power gains
using an external data source can only come at the price of type-I error
inflation \parencite[see][]{kopp-schneider_power_2020}. This shortcoming
is usually accepted as a lesser evil in the context of platform trials.
However, regulation may impose a constraint on maximal TOER inflation
per basket, which can be taken into account by implementing the very
harsh penalty \(\xi_3\).

\section{Optimization algorithms}
\label{sec:opt-algs}

For finding the optimal value (minimum or maximum) of a utility function
\(u(\cdot):\symbfupphi\mapsto u(\symbfupphi)\), we will consider the
optimization algorithms named in the following list. A brief explanation
of the functionality of each optimization algorithm is presented below
the list.

\begin{enumerate}
\item Bounded simulated annealing algorithm using the reflection approach for bounding the parameter space as suggested in \cite[Section 6 of][]{haario_simulated_1991} -- implemented in an R package developed for the purpose of this comparison study, using a start temperature of
\begin{enumerate}
\item $T_{\text{start}}=100$,
\item $T_{\text{start}}=10$,
\item $T_{\text{start}}=1$,
\end{enumerate}
and one function evaluation per temperature step.\footnote{Simulated annealing is an optimization algorithm inspired by the thermodynamic process of \enquote{annealing} in metal industry, hence the physical term \enquote{temperature}. Temperature is a transformation of the probability of jumping to a parameter state with lower utility with the hope of escaping local critical points in order to find global critical points. At the beginning of the run, temperature is relatively high and will then be decreased.}
\item Unbounded simulated annealing algorithm using the \enquote{return NA} approach for bounding the parameter space -- implemented in the R function \texttt{stats::optim()}, using the start temperature $T_{\text{start}}=10$ and one function evaluation per temperature step.
\item Differential evolution (DE) as implemented in the R package \texttt{metaheuristicOpt}, using a population size of 40, a scaling vector of 0.8 and a cross-over rate of 0.5.
\item Grey wolf optimizer (GWO) as implemented in the R package \texttt{metaheuristicOpt}, using a population size of 40.
\item  Constrained optimization by linear approximations (COBYLA) algorithm as implemented in the R package \texttt{nloptr}, using stopping tolerance of $10^{-6}$ in the parameter space of $\symbfupphi$ and a stopping tolerance of 0 in the value space of $u(\cdot)$.
\item Grid search algorithm, searching the set of all combinations of
\begin{align*}\lambda&\in\{0.2, 0.3, 0.4, 0.5, 0.6, 0.7, 0.8, 0.9\}\cup\{0.99\}\cup\{0.999\},\\
\varepsilon&\in\{0.0, 0.5, 1.0, 1.5, 2.0\}\cup\{5.0, 10.0, 15.0, 20.0, 25.0\},\text{ and}\\
\tau&\in\{0.0, 0.1, 0.2, 0.3, 0.4, 0.5, 0.6, 0.7, 0.8\}\cup\{1.0\}.\end{align*}
\end{enumerate}

We will allow each algorithm to run for up to 1000 function evaluations.
The grid of the grid search algorithm was chosen to use 1000 function
evaluations as well. In order to put focus on parts of the grid that
seemed most relevant in preliminary experiments while restricting to
grid dimensions \(10\times 10 \times 10\), we omitted some seemingly
less relevant parts: \(\lambda < 0.2\) (i.e.~test decision almost always
positive), high resolution for \(\varepsilon > 2\) (i.e.~sharing only
with high similarity) and \(\tau = 0.9\) (i.e.~almost no sharing,
similar to 0.8 and 1.0). If one of the simulated annealing algorithms,
GWO, DE, or the COBYLA algorithm shows poor convergence after 1000
evaluations and runtime permits it, we will conduct one more run with 20
000 function evaluations. With this number of iterations, preliminary
experiments with simulated annealing reached convergence, which were,
however, based on an R package unfit for the analysis.

We will end this chapter with a brief explanation of the functionality of
each optimization algorithm:

Simulated annealing, referred to as unbounded simulated annealing in
this publication, is a physics-inspired metaheuristic optimization
algorithm suggested by \cite{kirkpatrick_optimization_1983}. It builds
upon the Metropolis algorithm that is also used in Markov chain Monte
Carlo procedures. It is named simulated annealing as it mimics the
procedure of annealing in metallurgy: There, a metal product is heated
and then slowly cooled down in order to achieve a more homogeneous and
stable structure within the product. The algorithm starts at a
user-suggested or randomly selected initial parameter vector. In each
step, the simulated annealing algorithm randomly suggests a new
parameter vector near the old parameter vector. If the new vector has
better utility, The current parameter vector is updated to be the new
parameter vector. If it has worse utility, it is still replaced with a
probability proportional to current ``temperature''. At the beginning,
the temperature is high in order to allow the algorithm to escape local
optima. Following a prespecified temperature schedule, it is then slowly
cooled down in order to find the global optimum. This algorithm was
first shown to converge on finite parameter spaces in
\cite{geman_stochastic_1984}. Due to the finite nature of computer
memory, we can consider our parameter spaces as finite.

Bounded simulated annealing is a generalization of the simulated
annealing algorithm to hypercubes of the form
\(\prod_{i=1}^d [l_i, u_i]\) with lower and upper bounds \(l_i\) and
\(u_i\) considered as a subset of \(\mathbb R^d\). It works identically
to above-mentioned unbounded simulated annealing algorithm. Whenever a
suggested parameter vector's component \(\phi_i\) surpasses one boundary
\(l_i\) or \(u_i\), it is reflected along these boundaries until it lies
in the respective interval \([l_i, u_i]\). \cite{haario_simulated_1991}
suggested this modification and proved its convergence. This reflection
procedure's result can be calculated by a simple affine transformation
combined with division with remainder.

Constrained optimization by linear approximation (COBYLA) is an
optimization algorithm suggested by \cite{powell_direct_1994}. It
employs a \(d\)-dimensional simplex in order supply linear
approximations of the utility function without the need for calculating
derivatives. New simplex vertices are suggested by using these linear
approximations while dynamically adjusting search radius and punishment
for constraint violations. The algorithm's procedure is too complex to
be described in detail but can be found in above-mentioned reference.

Differential evolution (DE) is a metaheuristic optimization algorithm
that mimics genetic evolution
\parencite[see][for an overview]{das_differential_2011}. From a fixed
number of candidate vectors (``the population''), donor vectors are
generated by randomly adding the scaled differences of two vectors to a
third vector (``mutation''). Then, new candidates (``offspring'') are
generated by randomly replacing some vector components of the original
candidates with components of the donor vectors (``crossover'').
Finally, the new candidates replace the old candidates in the next
generation if they perform better or equal (``natural selection''). Many
improvements of this idea have been suggested as discussed by Das et al.

The grey wolf optimizer (GWO) is a metaheuristic optimization algorithm
inspired by the hunting behavior and social hierarchy of grey wolves
suggested by \cite{mirjalili_grey_2014}. The parameter space is searched
by a number of candidate vectors (``pack of wolves'') which are
following the direction of the three best candidates (``alpha, beta and
delta wolves''). Candidate vectors are allowed more random behavior in
the beginning (``searching for prey'') and are more strictly following
the best solutions in the end (``encircling the prey''). The algorithm
has further been improved
\parencite[see][]{nadimi-shahraki_improved_2021}, but this improved
version is not implemented in R.

Grid search is the conceptually simplest of the mentioned algorithms.
For each component of the parameter vector, the user specifies a set of
values of interest. Then, the algorithm simply searches all possible
combinations of values for the optimal value. Grid search can be
parallelized and is completely deterministic, but on the other hand, its
cost grows exponentially with the number of parameters. Its precision
will never be finer than the size of the mesh.

This choice of optimization algorithms is obviously not exhaustive of
the abundance of available optimization algorithms. We chose these
algorithms as they represent different approaches to optimization:
stochastic metaheuristics with inspirations from physics, genetics and
swarm behavior (simulated annealing, DE, and GWO, respectively),
non-linear programming (COBYLA), and naive deterministic approaches
(grid search). Both COBYLA and grid search were already applied to the
optimization of clinical trial designs, see \cite{kunzmann_adoptr_2021},
\cite{jiang_optimal_2021}, and \cite{baumann_basket_2023}. Stochastic
metaheuristics appear to be a good alternative as they have little
requirements to the \enquote{niceness} of the targeted utility function
and as they are often able to escape local minima. Availability of R
implementations was also a relevant criterion in the selection of
optimization algorithms.

\section{Comparison protocol}

The comparison study is divided into three parts with I. the goal of
comparing optimization algorithms, II. the goal of comparing utility
functions, and III. the goal of comparing the optimized parameter values
to the parameter values suggested in \cite{fujikawa_bayesian_2020} as
described above. For each part of the study, we will apply the ADEMP
scheme for describing simulation studies that was introduced in
\cite{morris_using_2019}. The ADEMP scheme was developed for describing
simulation studies of statistical methods. In Part I, the methods of
interest are optimization algorithms rather than statistical methods.
However, the ADEMP scheme could still be adapted to match the best
practices in benchmarking algorithms described in
\cite{beiranvand_best_2017}, namely
\enquote{clarifying the reason for benchmarking} (\emph{aim} in the
ADEMP scheme), selecting an appropriate test set
(\emph{data-generating mechanism} in ADEMP) and reporting comparative
measures of efficiency, reliability and quality of solution
(\emph{performance measures} in ADEMP). Some more sophisticated methods
from \cite{beiranvand_best_2017} such as the choice of an exhaustive
test problem set and the reporting of performance profile plots was
omitted as our algorithm comparison is a quite small case study rather
than a complete benchmarking of possible algorithms choices.

\subsection{Part I: Comparison of optimization algorithms}

In this first part of our comparison study, we will explore what
optimization algorithm is best suited for the utility-based optimization
approach. To this end, we will test the optimization algorithms on a
selection of utility functions and outcome scenarios. Judging from some
preliminary simulation attempts, it is expected that testing the
optimization algorithms on all utility functions will take too long to
be numerically feasible, see \Cref{sec:simsize} for a detailed
explanation.

\begin{enumerate}
    \item Aim: The goal of this part of the comparison study is to select the fastest among all reliable algorithms for optimizing the parameters of Fujikawa's basket trial.
    \item Test problems: We will consider the utility functions \emph{scenario-averaged two-level stratum-wise power-error function} $\bar u_ {\text{2ewp}}$ and \emph{scenario-averaged expected number of correct decisions function} $\bar u_{\text{ecd}}$ and will optimize the functions on one scenario set, namely the scenario set from \Cref{sec:scenarios} with $(I=4, n_i=20, p_0=0.15)$, yielding a total of two optimization test problems (two functions with one scenario set each). The deterministic algorithms will only be run once on each test problem, whereas the stochastic algorithms will be run $n_{\text{runs}}=50$ times on each test problem. The number of algorithm runs is justified in \Cref{sec:simsize} below. The seed of the first run will be 1856. As a start value for simulated annealing and COBYLA, we will choose $\symbfupphi_{\text{start}} = (\lambda_{\text{start}}, \varepsilon_{\text{start}}, \tau_{\text{start}})=(0.2, 0.5, 0)$, i.e. a test decision that is mostly positive and borrowing that takes place most of the time -- this will be suboptimal in most scenarios, as it very frequently commits a type-I error.
    \item Estimand/target: The target of each optimization algorithm is to find the optimal parameter combination with respect to a utility function as quickly as possible.
    \item Methods: We will compare the six optimization algorithms described in \Cref{sec:opt-algs}. Bounded simulated annealing will be tested with three different starting temperatures, resulting in a total of eight studied algorithms.
    \item Performance measures: The following comparative measures are of interest for comparing the different algorithms. Performance measures will be presented separately per test problem. If the performance measure was measured for each run of a stochastic algorithm, the measures will be summarized using mean, standard deviation, minimal and maximal values as appropriate. For selecting the best optimization algorithm, we will use the following approach: Across all test problems, we will calculate the mean performance measures. We will begin with internal reliability. Only optimization algorithms with an internal reliability of over 99\% will be considered for comparing external reliability. Only optimization algorithms with a success rate of over 99\% will be considered for comparing speed. Finally, the fastest of all the remaining optimization algorithms will be considered the best algorithm.
    \begin{enumerate}
        \item Efficiency: number of fundamental evaluations of the utility function, user CPU time, system CPU time, wall clock time, memory usage.
        \item Internal consistency: mean resulting optimal utility function value, marginal means of the optimal parameter vector components, the component's marginal sample standard deviations, the 95\%-confidence interval of the means assuming normality, minimal and maximal values of resulting function values and parameter vector components.
        \item External reliability: The true optimal solution is unknown, but the grid search algorithm will be used as a reference benchmark, as it is a deterministic algorithm that exhausts the whole parameter space, up to the grid's precision. The following performance measures will then be considered: success rate of delivering an optimal utility value greater than or equal to the grid search results, 95\%-confidence interval of the difference to the grid search result assuming normality, minimal and maximal difference to the grid search results.
    \end{enumerate}
    \item Reporting: We will provide tabular presentation of the performance measures. For the stochastic algorithms, Monte Carlo standard errors (MCSE) of the performance measures will be reported wherever estimating formulae of MCSE are known. In addition, we will use box plots to visualize performance measures of interest as appropriate. Furthermore, we will generate line plots showing the convergence of the simulated annealing runs, each line representing one run, with function evaluations on the x-axis and the utility function value or one parameter on the y-axis. Finally, we will generate a four-dimensional plot of the grid search run in order to visualize the shape of each utility function: the x-axis will represent the parameter $\varepsilon$, the y-axis the parameter $\tau$, plot facets will represent the parameter $\lambda$ and color will represent the utility function value $u(\lambda,\varepsilon,\tau)$.
\end{enumerate}

\subsection{Part II: Comparison of utility functions}

In the second part of the comparison study, we will explore which
utility function is best-suited for the optimization of Fujikawa's
basket trial.

\begin{enumerate}
\item Aim: The aim of this part of the comparison study is to find the utility function which achieves the best compromise between single-stratum power and EWP on the one hand and single-stratum TOER and FWER on the other hand.
\item Data: We use the fastest reliable algorithm found in Part I of the study to optimize the tuning parameters $\symbfupphi$ with respect to the utility functions of interest. If the algorithm is stochastic, a seed will be fixed to 899. The functions will be optimized for the seven scenario sets introduced in \Cref{sec:scenarios}. The scenario-averaged utility functions will be averaged across all scenarios in the respective scenario set. The single-scenario utility functions, i.e., $u_{\text{l}}(\symbfupphi,\mathbf{p})$ with $l=$ 2ewp or 2pow, will be optimized for the scenarios $\mathbf{p}=$ \enquote{2 of 3 active}, \enquote{2 of 4 active} and \enquote{4 of 8 active}, respectively. In the scenario sets with $(I=3, n_i=24, p_0=0.2)$ \parencite[from][]{fujikawa_bayesian_2020} and $(I=4, n_i=20, p_0=0.15)$ \parencite[from][]{baumann_basket_2023}, all utility functions can be calculated exactly up to the precision of numerical integration using the \texttt{baskexact} R package \parencite{baumann_baskexact_2023}. Hence, this part is not actually a simulation study. The number of runs to calculate the optimal results will be 1. However, for all other scenario sets, we rely on simulation for calculating the performance measures.
\item Estimands: The estimand of each utility function is the optimal parameter vector. By applying the optimization algorithm to the respective utility function, we will receive an optimal parameter vector. The optimal parameter vector should of course be optimal with respect to the respective utility function, but should also show satisfactory performance with respect to the performance measures mentioned below.
\item Methods: The methods of interest are the eight different utility functions mentioned in \Cref{sec:ut-fun}.
\item Performance measures: For $(I=3, n_i=24, p_0=0.2)$ and $(I=4, n_i=20, p_0=0.15)$, all of the performance measures mentioned below can be calculated exactly up to the precision of numerical integration using the \texttt{baskexact} R package, without the necessity of Monte Carlo simulation. For all other scenarios, performance measures will be calculated using Monte Carlo simulation as implemented in the \texttt{basksim} package. A parameter combination is always optimal \emph{with respect} to the scenario set, e.g. with respect to $(I=3, n_i=24,p_0=0.2)$. For each scenario from \Cref{sec:scenarios} the respective scenario set, the following performance measures will be reported:\footnote{Performance measures will even be reported if the optimal parameter combination was optimized for another scenario. For example, if $u_{\text{ewp}}(\symbfupphi,\mathbf{p_1},\mathbf{p_2})$ was optimized for $\mathbf{p_1}$ being the 2 of 4 strata active scenario and $\mathbf{p_2}$ being the global null scenario with four strata, then we will still report the performance measures for all scenarios from the set $(I=4, n_i=20, p_0=0.15)$ mentioned in \Cref{sec:scenarios}.} 
  \begin{enumerate}
  \item Marginal rejection rate of the local null hypothesis in each stratum, equivalent to TOER if the stratum is inactive and power if the stratum is active,
  \item FWER,
  \item EWP\footnote{The EWP is equal to the expected sensitivity in the terminology of \cite{krajewska_new_2021}. The FWER is equal to one minus the expected specificity in their terminology.},
  \item expected number of correct decisions,
  \item utility function value of all utility functions.
  Depending on the results, it may be difficult to select a clear \enquote{best choice} among the utility functions. Based on the results per scenario as well as pooled results across all scenarios, we will attempt to discuss advantages and disadvantages among the utility functions in order to suggest a \enquote{best practice}.
  \end{enumerate}
\item Reporting: Tabular reports of all performance measures will be provided. In addition, dot plots of the performance measures will be provided for each stratum count $I=3$, $4$ and $8$. On the x-axis, the scenarios will be sorted by the number of active strata followed by the scenarios with mixed true rates. On the y-axis, the respective performance measure will be plotted.
\end{enumerate}

\subsection{Part III: Comparison of optimal parameter combinations to Fujikawa's suggested parameter combination}

Part III is an addition to the methods studied in Part II. In addition
to the optimal parameter vectors obtained by utility optimization in
Part II, we will also calculate the same performance measures for the
parameter choice suggested in \cite{fujikawa_bayesian_2020}. There, the
shape parameter \(\varepsilon\), the similarity cutoff \(\tau\) and the
detection threshold \(\lambda\) are suggested to be either
\(\symbfupphi_{\text{Fuj (i)}} = (\lambda, \varepsilon, \tau) = (0.99, 2, 0)\)
or \(\symbfupphi_{\text{Fuj (ii)}} = (0.99, 2, 0.5)\).

\section{Further analyses}

In an exploratory fashion, we will consider two further aspects in our
comparison study. Firstly, Fujikawa's basket trial design could be
altered by replacing the Jensen-Shannon divergence \(\JSD\) by the
Hellinger distance \parencite[see e.g.][]{lecam_asymptotics_2000},

\[ \HLD(P, Q) = 1 - \int_{\mathcal X} \sqrt{P(x)Q(x)}\mu(dx),\]

which has the advantage that for two beta distributions, it can be
calculated from basic functions without the need for numerical
integration \parencite[see][]{sasha_answer_2012}:

\[ \HLD\left(\mathrm{Beta}(a_1,b_1),\mathrm{Beta}(a_2,b_2)\right) = 1 -\frac{B\left(\frac{a_1+a_2}{2},\frac{b_1+b_2}{2}\right)}{\sqrt{B(a_1,b_1)B(a_2,b_2)}}.\]

We will explore whether this replacement results in a speedup of the
design while at the same time maintaining a similar behavior compared to
Fujikawa's design.

Secondly, we will graphically investigate the effect of borrowing on the
maximal TOER in \(I = 2\) strata. One inactive basket will be kept at a
stable response rate of \(p_1 = 0.2\), while we increase the response
rate of the other basket from \(p_2 = 0.2\) to \(p_2 = 1\). We will then
plot the response rate of \(p_2\) on the x-axis and the TOER of basket 1
on the y-axis for different combinations of \(\varepsilon\) and
\(\tau\). The resulting curves may give an overview of how borrowing
affects the TOER in inactive strata.

\section{Justification of simulation size and runtime}
\label{sec:simsize}

For assessing the precision of the stochastic algorithms in Part I of
our comparison study, we decided to run the algorithms for a total of
\(n_{\text{runs}}=50\) times. In the following \Cref{sec:runtime}, we
will justify the choice of this simulation size by estimating the
study's run time. For the calculation of design operating
characteristics, we will use the \emph{baskexact} package, which
calculates exactly up to numerical imprecision whenever it is feasible.
However, we use the Monte Carlo-based R package \emph{basksim}
(\cite{baumann_basksim_2024}) for estimating design characteristics in
the case of large sample sizes and large stratum counts, as the
\emph{baskexact} package would take too long (see \Cref{sec:runtime} for
a more detailed explanation). We use \(n_{\text{MC}} = 1000\) for the
number of Monte Carlo-generated data sets in \emph{basksim}. The choice
of this number is justified in \Cref{sec:mcbasksim}.

For scenario sets \((I=3, n_i=24, p_0=0.2)\) and
\((I=4, n_i=20, p_0=0.15)\), Part II and Part III of our comparison
study are deterministic and precise up to the precision of numerical
integration. Therefore, we do not conduct several runs in these parts.
In other words, the simulation size for parts II and III equals 1. For
the other scenario sets, the calculation of performance measures is
based on Monte Carlo-generated data sets with \(n_{\text{MC}} = 1000\).

\subsection{Estimation of runtime}
\label{sec:runtime}

In order to estimate the duration of our simulation, we ran a small
pilot simulation: We executed the bounded simulated annealing algorithm
and the grid search with 1000 iterations each for optimizing the utility
function \(u_{\text{ewp}}(\symbfupphi,\mathbf{p_1}, \mathbf{p_2})\) with
respect to stratum counts \(I\), per-stratum sample sizes \(n_i\) and
the respective global null hypothesis \(\mathbf{p_2}\) using the R
packages \emph{baskexact} and \emph{basksim}. The resulting run times
are specified in \Cref{tab:runtime-pilot}.

\begin{table}[p]
\caption{Run time of 1000 iterations in a pilot study}
\centering
\begin{tabularx}{\textwidth}{lccXcr}
  \hline
\textbf{Algorithm} & $I$ & $n_i$ & \textbf{Scenario $\mathbf{p_1}$ }                                         & \textbf{Used package}         & \textbf{Run time }                  \\ 
  \hline
Simulated annealing & 3    & 24   & $(0.2, 0.2, 0.5)$                                                       & \emph{baskexact}   & $\SI{15.64}{\min}$ \\ 
                                 & 4    & 20  & $(0.15, 0.15, 0.4, 0.4)$                                            &\emph{baskexact}    & $\SI{80.24}{\min}$ \\ 
                                 & 8    & 15  & $(0.15, 0.15, 0.15, 0.15,\break 0.45, 0.45, 0.45, 0.45)$ & \emph{basksim}       & $\SI{60.12}{\min}$ \\ 
Grid search               & 3    & 24 & $(0.2, 0.2, 0.5)$                                                        & \emph{baskexact}   & $\SI{1.03}{\min}$  \\ 
                                 & 4    & 20 & $(0.15, 0.15, 0.4, 0.4)$                                             &\emph{baskexact}  & $\SI{4.52}{\min}$  \\
 & 8 & 15 & $(0.15, 0.15, 0.15, 0.15,\break 0.45, 0.45, 0.45, 0.45)$ & \emph{basksim}& $\SI{201.67}{\min}$\\ 
   \hline
\end{tabularx}
\label{tab:runtime-pilot}

\caption{Estimation of the total run time for Part I of the comparison study}
\begin{tabularx}{\textwidth}{lXl}
\toprule
Algorithms                            & Bounded simulated annealing with 3 starting temperatures, unbounded simulated annealing, DE, GWO (each of these six with $n_{\text{runs}}=50$ runs in order to investigate stochastic behavior), COBYLA, grid search &  \\
Utility functions                     & Scenario-averaged utility functions $\bar u_ {\text{2ewp}}$ and $\bar u_{\text{ecd}}$ \\
Scenarios                             & Seven scenarios with $(I=4, n_i=20, p_0=0.15)$                                                                                                                                                                                                          \\
Run time & $\SI{80.24}{\min}$ on a scenario with $(I=4, n_i=20, p_0=0.15)$ using the \emph{baskexact} package                                                                                                                                                         \\
Computation clusters                  & 20 computation kernels on the institute's RStudio server resulting in a speedup by a factor of 15                                                                                                                                                                                       \\
\midrule
Estimated total run time       & $(6\cdot n_{\text{runs}}+2)\cdot 2 \cdot 7 \cdot \SI{80.24}{\min} \cdot \frac{1}{15} = \SI{15.7}{\day}$      \\                                                                                   
\bottomrule
\end{tabularx}
\label{tab:runtime-part1}

\caption{Estimation of the total run time for Part II of the comparison study}
\begin{tabularx}{\textwidth}{lXl}

\toprule
Algorithms                            & Fastest reliable algorithm selected in Part I &  \\
Utility functions                     & Four single-scenario utility functions, four scenario-averaged utility functions, four scenario-averaged utility functions with maximal TOER penalty \\
Scenarios                             & One scenario per scenario set for the single-scenario utility functions; 4, 7, 9, 10, 11, 5, 4 scenarios per scenario set for the scenario-averaged utility functions  \\
Run time & $\SI{15.64}{\min}$ on a scenario with $(I=3, n_i=24, p_0=0.2)$, $\SI{80.24}{\min}$ on a scenario with $(I=4, n_i=20, p_0=0.15)$ using the \emph{baskexact} package,   $\SI{60.12}{\min}$ on a scenario with $(I=8, n_i=15, p_i=0.15)$ with the \emph{basksim} package, similar for large scenario sets                                                                                                                                                                                                                                                                                                           \\
Computation clusters                  & 1 computation kernel on the institute's RStudio server resulting in no speedup     \\
\midrule
Estimated total run time       & $1\cdot(4\cdot (\SI{15.64}{\min} + \SI{80.24}{\min} + \SI{60.12}{\min}) + (4+4)\cdot(4\cdot\SI{15.64}{\min} + 7\cdot\SI{80.24}{\min} + (9 + 10 + 11 + 5 + 4)\cdot\SI{60.12}{\min}))\cdot 1 = \SI{16.9}{\day}$      \\                                                                                   
\bottomrule
\end{tabularx}
\label{tab:runtime-part2}
\end{table}

The \emph{baskexact} package took \(\SI{80.24}{\min}\) to run for
\((I=4, n_i=20, p_0=0.15)\). This package calculates operating measures
combinatorically and the runtime is hence influenced by the number of
combinations, which is proportional to \((n_i)^I\). Therefore, it is
reasonable to assume that larger scenarios, e.g.~with
\((I=8, n_i=15, p_0=0.15)\) would have a longer duration by several
orders of magnitude more than \(\SI{80.24}{\min}\)
(\(\frac{15^8}{20^4}=c\cdot10^5\)), rendering any simulation infeasible.
This is the reason why we will calculate the utility function for the
other scenario sets with the R package \emph{basksim}, which applies
Monte Carlo simulation and can hence reduce run time. The execution of
\emph{basksim} used \(n_{\text{sim}}=1000\) Monte Carlo simulations in
each call to the package.

The grid search algorithm was faster than the simulated annealing
algorithm by more than a factor 15 even though both algorithms ran for
1000 iterations. This is due to the fact that we used parallelization on
20 workers on our institute's RStudio server for calculating the
performance measures on the grid. Due to the sequential stochastic
nature of metaheuristic optimization algorithms, parallelization of a
single simulated annealing, GWO or DE run is not possible. However, we
will be able to parallelize part of the \(n_{\text{runs}}=50\) simulated
annealing runs; hence we will also estimate a speedup by a factor 15 for
Part I. In \Cref{tab:runtime-part1}, we show a detailed explanation of
the run time of Part I. The estimated run time of 15.7 days is long.
However, it is still feasible while underlining the necessity to keep
the example set of utility functions as small as it is.

The run time of Part II of our study is shown in the following
\Cref{tab:runtime-part2}. The estimation of the benefit of
parallelization appears more complicated in this case: Should we
parallelize across the set of 12 utility functions or across the set of
up to 11 data scenarios? However, even without any parallelization, the
estimated run time is 16.9 days, which would be feasible.

Finally note that the grid search algorithm on \(I=8\) strata is slower
by a factor of 3 compared to simulated annealing. Grid search using the
\emph{basksim} package cannot use parallelization, as the \emph{basksim}
package already uses parallelization internally. However, it should
optimally have about the duration as the simulated annealing algorithm,
as both take the same number of iterations if executed sequentially. The
unwanted slowdown may be due to a suboptimal R implementation and will
be further investigated before execution of the comparison study.

\subsection{Monte Carlo standard error of the algorithms' precision in Part I}

In Part I of the comparison study, we are interested in the precision of
partly stochastic optimization algorithms. We use \(n_{\text{runs}}=50\)
to estimate this precision. How precise will our estimates be? The
marginal sample standard deviation of the components of the resulting
optimal parameter vectors is one of the most relevant performance
measures of Part I of the comparison study. Assuming a normal
distribution of the marginal optimal algorithm result components around
the true optimal vector components, we want to keep the standard
deviation of the marginal standard deviation of the vector components
below a reasonable level. An unbiased estimator of the standard error of
the sample standard deviation is given by

\[ \widehat{\mathrm{SD}}(s) = s \cdot \frac{\Gamma( \frac{n-1}{2} )}{ \Gamma(n/2) } \cdot \sqrt{\frac{n-1}{2} - \left( \frac{ \Gamma(n/2) }{ \Gamma( \frac{n-1}{2} ) } \right)^2 },\]

where \(s = \sqrt{ \frac{1}{n-1} \sum_{i=1}^{n} (X_i - \overline{X}) }\)
is the sample standard deviation
and\break \(\Gamma(z)=\int_0^\infty t^{z-1}e^{-t} dt\) is the gamma
function for \(z\in\mathbb{C}\) with positive real part
\(\mathfrak{R}(z)>0\), see Appendix, \Cref{sec:se-of-sd}, for details.
For a total of 50 algorithm runs, this means that we will achieve a
standard error of \(\widehat{\mathrm{SD}}(s)_{50} = 0.10127\cdot s\),
i.e., the sample standard deviation will be precise up to a standard
error of little more than 10\%. Regarding the already long run time,
this appears acceptable.

\subsection{Monte Carlo standard error of the design characteristics in parts I and II}
\label{sec:mcbasksim}

In parts I and II, we will estimate the design characteristics of basket
trial designs with \(I=8\) using the Monte Carlo-based package
\emph{basksim}. We want to estimate the stratum-wise power and TOER by
applying the design to \(n_{\text{MC}}=1000\) simulated data sets.
According to \cite[][Table 6]{morris_using_2019}, the Monte Carlo
standard error of a rejection rate estimate
\(\widehat{\mathrm{rate}}=\frac{1}{n_{\text{MC}}}\sum_{l=1}^{n_{\text{MC}}} \mathbf{1}(p_l\leq \alpha)\)
such as power and TOER is given by

\[ \widehat{\mathrm{SD}}(\widehat{\mathrm{rate}})=\sqrt{\frac{\widehat{\mathrm{rate}}\cdot(1-\widehat{\mathrm{rate}})}{n_{\text{MC}}}}.\]

The value is maximal for \(\widehat{\mathrm{rate}} = 0.5\), resulting in
a Monte Carlo standard error of
\(\widehat{\mathrm{SD}}(\widehat{\mathrm{rate}})\leq 0.016\) for
\(n_{\text{MC}}=1000\), i.e., a standard error of less than 2\%. This
seems acceptable regarding the long run time in Part I of our comparison
study.

\section{Discussion}

Utility functions are a feasible and objective way of combining
operating characteristics in clinical trials, which has proved useful in
different contexts. So far in the context of basket trials, optimization
is usually restricted to heuristic manual tuning of parameters as in
\cite{fujikawa_bayesian_2020} or to optimizing one characteristic of
interest (e.g.~expected number of correct decisions) while keeping
type-I error rate in one scenario controlled as in
\cite{broglio_comparison_2022} and \cite{baumann_basket_2023}. The
challenge of optimizing across multiple scenarios and the choice of
optimization algorithm is usually not discussed.
\cite{jiang_optimal_2021} suggested two types of utility functions but
their choice was also not compared to other functions. Hence, our
comparison study will fill a research gap in investigating both the
choice of utility functions and of optimization algorithms.

If it proves effective and feasible, the studied optimization framework
may lay the foundation for further research on optimizing basket trials,
be it extensions to unbalanced sample size in the strata, multi-stage
basket trials, and different Bayesian or frequentist basket trial
designs.

\section{Acknowledgements}
We would like to thank Lukas Baumann, Marietta Kirchner and Paul
Thalmann from the Institute of Medical Biometry at the Heidelberg
University Hospital, Germany, Norbert Benda from the German Federal
Institute for Drugs and Medical Devices (BfArM), and Carolin Herrmann
from Novo Nordisk, Denmark, for useful comments concerning this
protocol. Furthermore, we thank two anonymous reviewers for their
suggestions on how to improve the study protocol. This project is funded
by the German Research Organization (DFG) as part of the project
\emph{STOP OR GO}, grant KI 708/9-1.

\pagebreak
\printbibliography

\setcounter{section}{0}
\renewcommand\thesection{\Alph{section}}
\renewcommand\theHsection{A-\thesection}
\section{Appendix: Estimation of the standard error of the sample
standard deviation}
\label{sec:se-of-sd}
In the following, we present a calculation of the standard error of the sample standard deviation, following the explanation on StackExchange, \cite[see][]{macro_answer_2012}.

\begin{prop}
Consider $n$ random variables $X_i, i=1,\ldots,n$ which are independent identically normally distributed $X_i\sim N(\mu,\sigma^2)$. Furthermore, consider the usual consistent estimator of the sample standard deviation $s = \sqrt{\frac{1}{n-1}\sum_{i=1}^n(X_i-\bar X)}$ and let $\Gamma(z)=\int_0^\infty t^{z-1}e^{-t} dt$ be the gamma function for $z\in\mathbb{C}$ with positive real part $\mathfrak{R}(z)>0$. Then the following holds:
\begin{enumerate}
\item The expectation of the sample standard deviation is given by $E(s)=\frac{\sigma}{c_n}$ with the correction factor $c_n=\frac{\Gamma(\frac{n-1}{2})}{\Gamma(\frac{n}{2})}\sqrt{\frac{n-1}{2}}$.
\item $s_{\text{un}}=c_n\cdot s$ is a consistent and unbiased estimator of $\sigma$.
\item $\widehat{SD}(s)=s\cdot\sqrt{1-\frac{1}{(c_n)^2}}$ is a consistent estimator of $SD(s)$.
\item $\widehat{SD}_{\text{un}}(s)=c_n\cdot s\cdot\sqrt{1-\frac{1}{(c_n)^2}}= s\cdot\frac{\Gamma(\frac{n-1}{2})}{\Gamma(\frac{n}{2})}\sqrt{\frac{n-1}{2} - \left( \frac{\Gamma(\frac{n}{2})}{\Gamma(\frac{n-1}{2})}\right)^2}$.
\end{enumerate}
\end{prop}
\begin{proof}
The proof of 1. can be found in \cite{holtzman_unbiased_1950}. The proof of 2.\ is a straightforward corollary of 1.\ using the linearity of expectation for proving unbiasedness and the fact that $\lim_{n\to\infty}c_n=1$ \parencite[see][]{laforgia_further_1984} together with Slutsky's theorem for proving consistency.

Straightforward calculation using 1.\ shows that 
\[SD(s)=\sqrt{\mathrm{Var}(s)}=\sigma\cdot\sqrt{1-\frac{1}{(c_n)^2}}.\]
Using the fact that $s$ is a consistent estimator of $\sigma$ and the fact that $\lim_{n\to\infty}\sqrt{1-\frac{1}{(c_n)^2}}=0$, this calculation implies 3. Finally, 4.\ follows from 2.\ and 3 using linearity of expectation for proving unbiasedness and once again the fact that $\lim_{n\to\infty}\sqrt{1-\frac{1}{(c_n)^2}}=0$ for proving consistency.
\end{proof}

\paragraph*{Correspondence}
Lukas D Sauer, Institute of Medical Biometry, Heidelberg University, Im
Neuenheimer Feld 130.3, 69120 Heidelberg, Germany. E-mail:
\href{mailto:sauer@imbi.uni-heidelberg.de}{\nolinkurl{sauer@imbi.uni-heidelberg.de}}

\end{document}